\documentclass[11pt]{article}

\RequirePackage{etex}

\usepackage[margin=1in]{geometry}  	
\usepackage{amsfonts}
\usepackage{amsmath} 
\usepackage{booktabs}
\usepackage{multirow}
\usepackage[doublespacing]{setspace} 

\usepackage[svgnames]{xcolor}

\usepackage{float}
\usepackage{caption}
\DeclareCaptionLabelFormat{new-capT}{eAppendix Table}
\DeclareCaptionLabelFormat{new-capF}{eAppendix Figure}

\usepackage{graphicx}
\usepackage{adjustbox}
\usepackage{lscape}

\usepackage{adjustbox}

\usepackage{relsize}

\usepackage{bbm}
\usepackage{subfig}

\definecolor{forestgreen}{RGB}{34,139,34}
\usepackage{hyperref}
\hypersetup{
    colorlinks=true,
    linkcolor=magenta,
    filecolor=magenta, 
    citecolor=forestgreen,      
    urlcolor=blue
}

\usepackage{amsthm}
\newtheorem{theorem}{Theorem}
\newtheorem{proposition}[theorem]{Proposition}
\newtheorem*{remark}{Remark}

\usepackage{MnSymbol}
\usepackage{wasysym}

\newcommand{\graysquare}{\textcolor{gray}{\blacksquare}}
\newcommand{\blackdiamond}{\mathbin{\rotatebox[origin=c]{45}{$\blacksquare$}}}
\newcommand{\graydiamond}{\mathbin{\rotatebox[origin=c]{45}{$\graysquare$}}}

\usepackage{xpatch}
\makeatletter
\xpatchcmd{\proof}{\@addpunct{.}}{\@addpunct{:}}{}{}
\makeatother

\makeatletter
\newcommand{\vast}{\bBigg@{3}}
\newcommand{\Vast}{\bBigg@{4}}
\makeatother

\makeatletter
\newcommand*{\indep}{%
  \mathbin{%
    \mathpalette{\@indep}{}%
  }%
}
\newcommand*{\nindep}{%
  \mathbin{
    \mathpalette{\@indep}{\not}
  }%
}
\newcommand*{\@indep}[2]{%
  \sbox0{$#1\perp\m@th$}
  \sbox2{$#1=$}
  \sbox4{$#1\vcenter{}$}
  \rlap{\copy0}
  \dimen@=\dimexpr\ht2-\ht4-.2pt\relax
  \kern\dimen@
  {#2}
  \kern\dimen@
  \copy0 
} 
\makeatother

\DeclareMathOperator{\E}{\textnormal{\mbox{E}}}

\usepackage[utf8]{inputenc}
\DeclareUnicodeCharacter{200E}{}

\usepackage{tikz}
\usetikzlibrary{positioning,shapes.geometric}

\usepackage{tkz-graph}
\usepackage{pgf}
\usepackage{tikz}
\usetikzlibrary{arrows,shapes.arrows,shapes.geometric,
  shapes.multipartb,backgrounds,decorations.pathmorphing,positioning,fit,automata}
\tikzset{
  >=stealth',
  true/.style={
    rectangle,
    draw=black, very thick,
    text width=6.;/ mem,
    minimum height=2em,
    text centered,
    fill=gray, opacity = 0.5},
  punkt/.style={
    rectangle,
    rounded corners,
    draw=black, very thick,
    text width=6.5em,
    minimum height=2em,
    text centered},
  est/.style={
    circle,
    draw=black, very thick,
    text centered},
  shade/.style={
    circle,
    draw=black, very thick, fill=gray!50,
    text centered},
  weight/.style={
    circle,
    draw=black, very thick,
    text width=6.5em,
    minimum height=2em,
    text centered},
  pil/.style={
    ->,
    thick,
    shorten <=2pt,
    shorten >=2pt,},
  double/.style={
    <->,
    thick,
    shorten <=2pt,
    shorten >=2pt,},
  dash/.style={
    dashed,
    thick,
    shorten <=2pt,
    shorten >=2pt,},
  dashdouble/.style={
    <->,
    dashed,
    thick,
    shorten <=2pt,
    shorten >=2pt,}
}

\usepackage{setspace}
\usepackage{cite}

\usepackage[stable]{footmisc}

\usepackage{rotating}

\makeatletter
\def\@seccntformat#1{\@ifundefined{#1@cntformat}%
   {\csname the#1\endcsname\quad}  
   {\csname #1@cntformat\endcsname}
}
\let\oldappendix\appendix 
\renewcommand\appendix{%
    \oldappendix
    \newcommand{\section@cntformat}{\appendixname~\thesection\quad}
}
\makeatother

\usepackage{authblk}

\usepackage[absolute,showboxes]{textpos}

\TPMargin{5pt}

\newcommand{\copyrightstatement}{
    \begin{textblock}{0.84}(0.08,0.93)    
         \noindent
         \footnotesize
         This DRAFT manuscript presents WORK IN PROGRESS. \\
         Comments and reports of mistakes are very much welcome at \href{mailto:issa\_dahabreh@brown.edu}{issa\_dahabreh@brown.edu}.
    \end{textblock}
}

\usepackage{datetime}
\def\paperversionmajor{24}
\def\paperversionminor{0}

\begin{document}

\title{Towards causally interpretable meta-analysis: transporting inferences from multiple randomized trials to a new target population}

\author[1-3]{Issa J. Dahabreh}
\author[4]{Lucia C. Petito}
\author[1]{Sarah E. Robertson}
\author[3,5,6]{Miguel A. Hern\'an}
\author[7]{Jon A. Steingrimsson}

\affil[1]{Center for Evidence Synthesis in Health and Department of Health Services, Policy \& Practice, School of Public Health, Brown University, Providence, RI}
\affil[2]{Department of Epidemiology, School of Public Health, Brown University, Providence, RI}
\affil[3]{Department of Epidemiology, Harvard T.H. Chan School of Public Health, Boston, MA}
\affil[4]{Department of Preventive Medicine, Feinberg School of Medicine, Northwestern University, Chicago, IL}
\affil[5]{Department of Biostatistics, Harvard School of Public Health, Boston, MA}
\affil[6]{Harvard-MIT Division of Health Sciences and Technology, Boston, MA}
\affil[7]{Department of Biostatistics, School of Public Health, Brown University, Providence, RI}

\copyrightstatement

\maketitle{}
\thispagestyle{empty}

\clearpage
\noindent
\textbf{Type of manuscript:} Original research 

\vspace{0.1in}
\noindent
\textbf{Running head:} Causally interpretable meta-analysis 

\vspace{0.1in}
\noindent
\textbf{Conflicts of interest:} None declared

\vspace{0.1in}
\noindent
\textbf{Acknowledgments:} The data analyses in our paper used HALT-C data obtained from the National Institute of Diabetes and Digestive and Kidney Diseases (NIDDK) Central Repositories. This paper does not necessarily reflect the opinions or views of the HALT-C study, the NIDDK Central Repositories, or the NIDDK.

\vspace{0.1in}
\noindent
\textbf{Sources of funding:} This work was supported in part by Patient-Centered Outcomes Research Institute (PCORI) awards ME-1306-03758 and ME-1502-27794 (Dahabreh); National Institutes of Health (NIH) grant R37 AI102634 (Hern\'an); and Agency for Healthcare Research and Quality (AHRQ) National Research Service Award T32AGHS00001 (Robertson). The content of this paper does not necessarily represent the views of the PCORI, its Board of Governors, the Methodology Committee, the NIH, or AHRQ.

\vspace{0.1in}
\noindent
\textbf{Computer code:} We have provided \texttt{R} code to implement the methods described in this paper.

\thispagestyle{empty}

\clearpage
\setcounter{page}{1}

\vspace*{0.25in}

\begin{abstract}
\noindent
\linespread{1.7}\selectfont
We take steps towards causally interpretable meta-analysis by describing methods for transporting causal inferences from a collection of randomized trials to a new target population, one-trial-at-a-time and pooling all trials. We discuss identifiability conditions for average treatment effects in the target population and provide identification results. We show that assuming inferences are transportable from all trials in the collection to the same target population has implications for the law underlying the observed data. We propose average treatment effect estimators that rely on different working models and provide code for their implementation in statistical software. We discuss how to use the data to examine whether transported inferences are homogeneous across the collection of trials, sketch approaches for sensitivity analysis to violations of the identifiability conditions, and describe extensions to address non-adherence in the trials. Last, we illustrate the proposed methods using data from the HALT-C multi-center trial.\\


\end{abstract}

\clearpage
\section*{Introduction}
\setcounter{page}{1}

For many causal questions, we have effect estimates from several randomized trials that recruit individuals from different underlying populations, typically by convenience sampling. We would then like to synthesize the trial evidence and transport inferences to some target population that is chosen on substantive grounds \cite{dahabreh2019commentaryonweiss}. 

Studies that synthesize findings across multiple trials are known as ``meta-analyses'' \cite{cooper2009handbook}. Most meta-analyses combine the trial-specific effect estimates to obtain a pooled effect estimate using either random effects models or common effect models (often referred to as ``fixed effect'' models) \cite{laird1990some, higgins2009re, rice2018re}. A commonly overlooked problem, however, is that standard meta-analyses may produce results that do not have a clear causal interpretation when each trial includes individuals from a different population and the treatment effect varies across those populations. 

As an example, consider the effect of angiotensin-converting enzyme (ACE) inhibitors, compared with other anti-hypertensive agents, in individuals with nondiabetic chronic kidney disease \cite{giatras1997, jafar2001angiotensin, jafar2003}. A meta-analysis of 11 trials found that ACE inhibitors improved renal function and reduced progression to end-stage renal disease; the benefits of ACE inhibitors were substantially greater among individuals with high urine protein excretion than those with lower excretion \cite{jafar2003, schmid2004meta}. Across all trials, only 6\% of participants were African American and in only 34\% of participants the cause of kidney disease was hypertensive nephrosclerosis. Now suppose that we are interested in the effect of ACE inhibitors in a target population of non-diabetic African American adults with chronic kidney disease \cite{wright1996design, wright2002effect}. In African Americans, chronic kidney disease is most often due to hypertension; and proteinuria is more common compared with whites \cite{hsu2003racial, mcclellan2011albuminuria}. In addition, blood pressure in hypertensive African Americans is generally considered to respond better to calcium channel blocker or diuretic monotherapy, rather than ACE inhibitor treatment \cite{mann2014antihypertensive}. Given these differences and the under-representation of African Americans in the trials (in fact, 7 of the 11 trials did not enroll any African Americans), the estimate from a conventional meta-analysis of the 11 trials would have limited relevance for the target population.

In fact, whenever treatment effects are heterogeneous over variables that vary in distribution across trials (as is the case in the ACE inhibitor example) standard meta-analyses of treatment effects from the trials do not produce estimates with a clear causal interpretation for any reasonable target population. The problem arises because meta-analyses combine trial-specific effect estimates using weights that reflect the precision of the estimates (and their variability, for random effects meta-analyses) rather than their relevance to the target population \cite{eddy1992meta}. In fact, most meta-analyses do not even specify their target population, regardless of whether the meta-analysis is based on summary statistics or on individual participant data from each trial. Similar concerns about conventional meta-analysis methods have been recently summarized, independently of our work, in \cite{manski2019meta}.

In this paper, we take steps towards causally interpretable meta-analysis by proposing methods for extending inferences about average treatment effects from a collection of randomized trials to a target population, one-trial-at-a-time and pooling all trials. Our approach requires individual participant data from the randomized trials together with baseline covariate data from a random sample from the target population. The latter allows us to account for differences in the distribution of effect modifiers between the trial-specific populations and the target population. Besides describing identification results, we show that assuming inferences are transportable from all trials in the collection to the same target population has implications for the law underlying the observed data. We propose average treatment effect estimators that rely on different working models and provide code for their implementation in statistical software.  We discuss how to use the data to examine whether transported inferences are homogeneous across the collection of trials, sketch approaches for sensitivity analysis to violations of the identifiability conditions, and describe extensions to address non-adherence in the randomized trials. Last, we illustrate the proposed methods using data from the HALT-C multi-center trial.

\section*{Causal quantities of interest}

Suppose we have individual-level data from participants in a collection of randomized trials $\mathcal S = \{ 1, \ldots, m \}$. For simplicity, we assume that all trials estimate the effect of assignment to treatment $Z$ that takes values in the finite set $\mathcal Z$ on an outcome $Y$ measured at some fixed follow-up time (we do not consider failure-time outcomes). If some treatments of interest are only assessed in a subset of the available trials, we can introduce treatment pair-specific subsets of the collection $\mathcal S(z,z^\prime) \subseteq \mathcal S$ that consist of those trials in $\mathcal S$ that compared treatments $z$ and $z^\prime$. With this change, our results can be extended to so-called ``network meta-analyses'' \cite{lumley2002network}, in which not all treatments are assessed in all trials.

For each participant in each of the trials in the collection $\mathcal S$, we have information on the trial $S$ in which they participated, treatment assignment $Z$, baseline covariates $X$, and outcome $Y$. Therefore, for each trial $s\in \mathcal S$, the data consist of realizations of independent random tuples $(X_i, S_i = s, Z_i, Y_i)$, $i = 1, \ldots, n_s$, where $n_s$ denotes the total number of randomized individuals in trial $s$. 

We also obtain a simple random sample from the target population of interest. The individuals in this sample are not participating in any of the trials (either because they were not invited or were invited but declined to participate). We collect baseline covariate data from them, but need not collect treatment or outcome data; we discuss the relationship between the covariate distribution in the target population and the collection of trials in the \emph{Identification} section. We use the convention that $S = 0$ for the target population so the data from the sample of the target population consist of realizations of random tuples $(X_i, S_i = 0)$, $i = 1, \ldots, n_0$, where $n_0$ is the total number of sampled individuals. Throughout, we use $I(\cdot)$ to denote the indicator function; for example, $I(S \in \mathcal S)$ is a random variable that denotes participation in any of the trials in the collection $\mathcal S$, such that $I(S \in \mathcal S)$ = 1 when $S \in \mathcal S$; and 0, otherwise. 

We form a composite dataset by appending the data from all the trials and the sample from the target population. The dataset consists of independent realizations of the random tuple $\big (X_i, S_i, I(S_i \in \mathcal S) \times Z_i, I(S_i \in \mathcal S) \times Y_i \big)$, $i=1,\ldots,n$, where $n$ is the total number of observations in the trials and the sample from the target population, $n = \sum_{s=0}^{m} n_s$. Throughout, we use italicized capital letters to denote random variables and lower case letters to denote the corresponding realizations. We generically use the notation $f(\cdot)$ to denote densities. 

We discuss details of the underlying sampling model in eAppendix \ref{appendix:sampling_model}. Briefly, we assume that the observed data are obtained by random sampling from an infinite superpopulation of individuals that is stratified by $S$, with sampling fractions that are unknown and possibly unequal constants for each stratum with $S = s$, for $s \in \{0,1,\ldots, m\}$. We refer to our model as a biased sampling model \cite{bickel1993efficient} because the proportion of individuals in the data with $S = s$, for $s \in \{0,1,\ldots, m\}$, does not in general reflect the population probability of $S = s$, but is instead shaped by the complex process that drives the design and conduct of actual randomized trials (see references \cite{dahabreh2019transportingStatMed, dahabreh2019studydesigns, dahabreh2019efficient} for additional details in the single trial case).

Let $Y^z$ denote the potential (counterfactual) outcome under treatment assignment $z \in \mathcal Z$ \cite{rubin1974, robins2000d}. We are interested in estimating the average treatment effect of treatment assignment (i.e., the intention-to-treat effect) in the target population, $\E [Y^z - Y^{z^\prime} | S = 0]$ for every pair of treatments $z$ and $z^\prime$ in $\mathcal Z$.

\section*{Identification of average treatment effects in the target population using a single trial}\label{sec:identification}

We now consider transporting inferences from each trial to the target population. Derivations for all identification results reported in this section are provided in eAppendix \ref{appendix:identification}. The derivations are valid under the biased sampling model described in eAppendix \ref{appendix:sampling_model} and thus the expectations and probabilities below can be taken over the data law of that model.

\subsubsection*{Identifiability conditions} 

The following are sufficient conditions for identifying the average treatment effect in the target population, $\E[Y^z - Y^{z^\prime} | S = 0]$, using covariate, treatment, and outcome data from a single trial $s^* \in \mathcal S$ and baseline covariate data from the sample of the target population.

\vspace{0.1in}
\noindent
\emph{A1. Consistency of potential outcomes:} If $Z_i = z,$ then $Y^z_i = Y_i$, for every individual $i$ in trial $s^*$ or the target population. Implicit in this notation is the absence of any effect of trial engagement on the outcome \cite{dahabreh2019identification, hernan2011compound} (e.g., no Hawthorne effects).

\vspace{0.1in}
\noindent
\emph{A2. Conditional exchangeability over treatment assignment $Z$ in the trial:} $\E[Y^z | X = x, S =s^* , Z = z ] = \E[Y^z | X = x, S = s^*]$ for each treatment $z \in \mathcal Z$, and each $x$ with $f(x, S = s^*) > 0$. This condition is expected to hold because treatment assignment in each trial is randomized (possibly conditional on $X$).

\vspace{0.1in}
\noindent
\emph{A3. Positivity of the treatment assignment probability in the trial:} For each treatment $z \in \mathcal Z$, $\Pr[Z = z | X = x, S = s^*] > 0$ for every $x$ with $f(x, S  = s^*) > 0$. This condition is also expected to hold because of randomization.

\vspace{0.1in}
\noindent
\emph{A4. Conditional exchangeability in measure between the trial and the target population:} For each pair of treatments $z$ and $z^\prime$ in $\mathcal Z$, $\E[Y^z - Y^{z^\prime} | X = x, S = 0] = \E[Y^z  - Y^{z^\prime} | X = x, S = s^*]$ for every $x$ with $f(x, S =0)>0$. This condition encodes an assumption of no residual effect measure modification by trial participation, conditional on the baseline covariates $X$. It is mathematically weaker than the identifiability conditions used in most previous work on transportability (e.g., \cite{westreich2017}), but it is sufficient to transport inferences about the average treatment effects from trial $s^*$ to the target population \cite{dahabreh2019transportingStatMed}. Exchangeability in measure may be implied by distributional independence conditions, for example, $Y^z \indep I(S = s^*) | X, I(S \in \{0, s^*\}) = 1$, but exchangeability in measure can hold even if such independence conditions do not. Thus, the condition cannot be verified solely using graphical methods such as d-separation-based criteria applied to single world intervention graphs \cite{dahabreh2019identification} or selection diagrams \cite{pearl2014}.

\vspace{0.1in}
\noindent
\emph{A5. Positivity of the probability of participation in the trial:} $\Pr[S = s^* |X = x] > 0$ for every $x$ with $f(x, S = 0) > 0$. Informally, this condition means that, for covariates needed to establish conditional exchangeability in measure, each covariate pattern in the target population should have positive probability of occurring in the trial $s^*$. 

In stating the identifiability conditions, we have used $X$ generically to denote baseline covariates. It is possible however, that strict subsets of $X$ are adequate to satisfy the exchangeability condition for trial $s^*$ and the target population, or conditional exchangeability in measure. For example, if trial $s^*$ is marginally randomized, the mean exchangeability over treatment assignment $A$ among trial participants will hold unconditionally. Also, the identification results we obtain under these conditions apply beyond randomized trials, for example, to pooled analyses of observational studies comparing interventions \cite{blettner1999traditional}, provided that conditions \emph{A2} and \emph{A3} can be assumed to hold. In this setting, condition \emph{A2} corresponds to the usual ``no unmeasured confounding'' assumption \cite{hernan2020}, applied to study $s^*$. 

To focus on issues related to selective trial participation, we assume complete adherence to the trial protocol and no loss-to-follow-up, so that the intention-to-treat effect (i.e., the average effect of treatment assignment) is equal to the per-protocol effect (i.e., the average effect of receiving treatment as indicated in the protocol). If adherence to the protocol is incomplete, the two effects are not equal. The methods we describe here can be extended to account for non-adherence to the trial protocol, but the extensions do not offer additional insights regarding transportability from a collection of trials to a target population. As an example, in eAppendix \ref{appendix:non_adherence}, we demonstrate how the methods can be extended to address incomplete adherence in a two-period study. In the remainder of the main text, for brevity, we use the term ``average treatment effect'' to mean the intention-to-treat average effect in the target population ($S=0$).

\subsubsection*{Identification} 

Under conditions \emph{A1} through \emph{A5}, when transporting inferences comparing treatments $z$ and $z^\prime$ from a trial $s^* \in \mathcal S$ to the target population, the average treatment effect $\E[Y^z - Y^{z^\prime} | S = 0 ]$ equals the following functional of the observed data distribution:
\small
\begin{equation} \label{eq:identification_g_single}
 \psi(z, z^{\prime},s^*) \equiv \E \big[  \E[Y | X, S = s^*, Z = z ] - \E[Y | X, S = s^*, Z = z^\prime ]  \big| S = 0 \big].
\end{equation}
\normalsize

Under the positivity conditions, we can obtain an expression for $\psi(z, z^{\prime},s^*)$ that uses weighting,
\footnotesize
\begin{equation}\label{eq:identification_w_single}
    \begin{split}
  \psi(z, z^{\prime}, s^*) &= \dfrac{1}{\Pr[S = 0 ]}  \E\Bigg[  \left( \dfrac{I(Z = z )}{\Pr[Z = z | X, S = s^*]} - \dfrac{I(Z = z^\prime )}{\Pr[Z = z^\prime | X,S  = s^*]}  \right)  \dfrac{ I(S=s^*) Y \Pr[S = 0| X, I(S \in \{0, s^*\})= 1] }{\Pr[ S = s^*| X, I(S \in \{0, s^*\})= 1]  }  \Bigg].
  \end{split}
\end{equation}
\normalsize

The identification results in this section only involve the target population $S = 0$ and the trial $S = s^*$; thus, in addition to holding under the sampling model of eAppendix \ref{appendix:sampling_model} they also hold under the sampling model described in reference \cite{dahabreh2019studydesigns}.

\section*{Identification when pooling the trials}\label{sec:identification_pooled}

We now consider transporting inferences to the target population when using the data from all trials in the collection $\mathcal S$. Derivations for all identification results reported in this section are also provided in eAppendix \ref{appendix:identification} and are valid under the biased sampling model described in eAppendix \ref{appendix:sampling_model}.

\subsection*{Identifiability conditions} 

The following are sufficient conditions for identifying the average treatment effect in the target population, $\E[Y^z - Y^{z^\prime} | S = 0]$, using covariate, treatment, and outcome data from every trial in the collection $\mathcal S$ and baseline covariate data from the sample of the target population.

\vspace{0.1in}
\noindent
\emph{B1. Consistency of potential outcomes:} If $Z_i = z,$ then $Y^z_i = Y_i$, for every individual $i$ in the target population or the populations underlying the trials in $\mathcal S$. As noted before, implicit in this notation is the absence of any effect of trial engagement on the outcome.

\vspace{0.1in}
\noindent
\emph{B2. Conditional exchangeability over treatment assignment $Z$:} $\E[Y^z | X = x, S =s , Z = z ] = \E[Y^z | X = x, S = s]$ for every trial $s \in \mathcal S$, each treatment $z \in \mathcal Z$, and every $x$ with $f(x, S = s) > 0$.

\vspace{0.1in}
\noindent
\emph{B3. Positivity of the treatment assignment probability in the trials:} For every treatment $z \in \mathcal Z$, $\Pr[Z = z | X = x, S = s] > 0$ for every trial $s \in \mathcal S$ and every $x$ with $f(x, S  = s) > 0$.

\vspace{0.1in}
\noindent
\emph{B4. Conditional exchangeability in measure between the trial and the target population:} For every pair of treatments $z$ and $z^\prime$ in $\mathcal Z$, $\E[Y^z - Y^{z^\prime} | X = x, S = 0] = \E[Y^z  - Y^{z^\prime} | X = x, S = s]$ for every trial $s \in \mathcal S$ and every $x$ with $f(x, S =0)>0$.

\vspace{0.1in}
\noindent
\emph{B5. Positivity of the probability of participation in the trials:} $\Pr[S = s |X = x] > 0$ for every $s \in \mathcal S$ and every $x$ with $f(x, S = 0) > 0$.

\subsection*{Observed data implications}

Suppose now that conditions \emph{B1} through \emph{B5} hold, so that inferences are transportable from each trial $s \in \mathcal S$ to the target population; if effects are transportable only for a subset of the trials in the collection $\mathcal S$ it is easy to modify the results given  here to restrict them to the subset of $\mathcal S$ that satisfies conditions \emph{B1} through \emph{B5}. We will argue that conditions \emph{B1} through \emph{B5}, have important implications for law underlying the observed data. To see this, note that when conditions \emph{B4} and \emph{B5} hold, then for the pair of treatments $z, z^\prime$, and for all $X$ values in the common support with $S = 0$,
\begin{equation*}
    \E[Y^z - Y^{z^\prime}| X, S= 1] = \ldots = \E[Y^z - Y^{z^\prime} | X, S = m] = \E[Y^z - Y^{z^\prime} | X, S = 0].
\end{equation*}
That is, the conditional causal mean difference of each of the $m$ trials in the collection $\mathcal S$ is equal to the conditional causal mean difference in the target population. This means that the conditional average causal effect is independent of study participation in $\mathcal S$, within strata of baseline covariates.  

Using the above result and assumptions \emph{B1} through \emph{B3}, for the pair of treatments $z, z^\prime$ and for all $X$ values in the common support with $S = 0$, we obtain
\begin{equation}\label{eq:chain_observed_data_implications}
  \begin{split}
    &\E[Y | X, S= 1, Z = z] - \E[Y | X, S= 1, Z = z^\prime] = \ldots = \\
    &\quad\quad\quad\quad\quad\quad \E[Y | X, S = m, Z = z] - \E[Y | X, S= m, Z = z^\prime] \equiv \tau(z,z^\prime; X),
  \end{split}
\end{equation}
where we use the notation $\tau(z,z^\prime; X)$ for the common conditional mean difference. 

The above chain of equalities is an observed data implication of conditions \emph{B1} through \emph{B5} because it only involves observed variables. Thus, using the data to evaluate the chain of equalities can be viewed as a falsification test for whether the chain of assumptions \emph{B1} through \emph{B5} fails to hold for the entire collection $\mathcal S$. Nevertheless, such assessment cannot prove that the the assumptions hold for the collection or for any subset of trials. Substantive knowledge should be used to determine whether estimates of the conditional mean differences in the above equation are ``close enough'' across trials. Expert assessment may be aided by statistical testing, because the chain of equalities can be viewed as a null hypothesis to be tested using parametric or non-parametric methods \cite{delgado1993testing, neumeyer2003nonparametric, racine2006testing, luedtke2019omnibus}.

\subsection*{Identification}

Suppose that substantive knowledge suggests that assumptions \emph{B1} through \emph{B5} hold, so that all trials in the collection $\mathcal S$ are transportable to the target population, and that examination of the observed data implications of assumptions \emph{B1} through \emph{B5} does not raise any concern. Then, the average treatment effect $\E[Y^z - Y^{z^\prime} | S = 0 ]$ equals the following functional of the observed data distribution:
\small
\begin{equation} \label{eq:identification_g_collection}
 \phi(z, z^{\prime}) \equiv \E \big[  \tau(z,z^\prime; X)  \big| S = 0 \big].
\end{equation}
\normalsize
Under the positivity conditions, we can obtain an expression for $\phi(z, z^{\prime})$ that uses weighting,
\begin{equation}\label{eq:identification_w_collection}
    \begin{split}
  \phi(z, z^{\prime}) &= \dfrac{1}{\Pr[S = 0 ]}  \E\Bigg[  \omega(z, z^\prime; X, S, Y)  \dfrac{ Y \Pr[S = 0| X] }{\Pr[ I(S \in \mathcal S) = 1| X]  } \Bigg],
  \end{split}
\end{equation}
where $$\omega(z, z^\prime; X, S) =  \left( \dfrac{I(Z = z )}{\Pr[Z = z | X, S, I(S \in \mathcal S) = 1]} - \dfrac{I(Z = z^\prime )}{\Pr[Z = z^\prime | X,S, I(S \in \mathcal S) = 1]}  \right) I(S \in \mathcal S).$$

\subsection*{Relaxing positivity condition \emph{B5} when pooling the trials}

Intuition suggests that by combining information across multiple trials it should be possible to relax positivity condition \emph{B5}, which, as stated above, requires sufficient overlap between the target population and each one of the trials in the collection $\mathcal S$. 

As an example, consider two randomized trials, comparing the same anti-hypertensive treatments, one enrolling patients with mild and the other with severe hypertension, and a target population, $S=0$, that includes individuals with both mild and severe hypertension. It should be possible to obtain inferences about the target population provided one is willing to assume that (1) conditional average treatment effects from each trial are equal to conditional average treatment effects in the corresponding subset of the target population; and (2) every covariate pattern that can occur among individuals in the target population with mild hypertension has positive probability of occurring in the trial of mild hypertension and every covariate pattern that can occur among individuals in the target population with severe hypertension has positive probability of occurring in the trial of severe hypertension.

Using $\mathcal X_s$ to denote the support of the distribution of $X$ given $S = s$, the second condition above can be generalized as follows: the union of the intersections of the support sets of the distribution of $X$ in each trial in $\mathcal S$ and the target population equals the support set of the distribution of $X$ in the target population. Using $\mathcal X_s$ to denote the support of the distribution of $X$ given $S = s$, this more general overlap condition can be written concisely as $\bigcup\limits_{s \in \mathcal S} \mathcal X_s \cap \mathcal X_0 = \mathcal X_0$. Though identification is intuitively obvious in simple cases like the two trial hypertension example, handling identification in general cases when condition \emph{B5} is violated but the overlap condition $\bigcup\limits_{s \in \mathcal S} \mathcal X_s \cap \mathcal X_0 = \mathcal X_0$ holds, requires modifications to condition \emph{B4} and the introduction of additional notation to keep track of the subsets of $\mathcal S$ where each covariate pattern $X = x$ can occur. For that reason, we we do not pursue it here.

\section*{Estimation \& Inference}

\subsection*{Estimation for transporting individual trials}

We now use the identification results on the previous section to propose estimators of average treatment effects; when transporting results from a single randomized trial, these estimators relate to the potential outcome mean estimators of \cite{dahabreh2019transportingStatMed}.

\vspace{0.1in}
\noindent
\emph{Estimation by modeling the conditional average treatment effect:} The first option is to use an estimator based on (\ref{eq:identification_g_single}),
\begin{equation}\label{eq:estimation_g_single}
  \widehat \psi_{\text{\tiny te}}(z, z^\prime, s^*) = \Bigg\{\sum\limits_{i=1}^n I(S_i = 0) \Bigg\}^{-1} \sum\limits_{i=1}^{n} I(S_i = 0)  \widehat b(z, z^\prime,s^*; X_i),
\end{equation}
where $\widehat b(z, z^\prime, s^*; X)$ is an estimator of $\E[Y | X, S = s^*, Z = z] - \E[Y | X, S = s^*, Z = z^\prime]$. A simple way to estimate $\widehat b(z, z^\prime,s^*; X)$ is to use an outcome model-based estimator $\widehat h(z, s^*; X)$ of $\E[Y | X, S = s^*, Z = z]$, for each $z \in \mathcal Z$, and then estimate $\widehat b(z, z^\prime, s^*; X)$ by taking the difference $\widehat h(z,s^*; X) - \widehat h(z^\prime,s^*; X)$. Alternatively, because under assumptions \emph{A1} through \emph{A3}, $\E[Y | X, S = s^*, Z = z] - \E[Y | X, S = s^*, Z = z^\prime] = \E[Y^z - Y^{z^\prime} | X, S = s^*]$, we can use a g-estimator of the conditional average treatment effect in trial $s^*$ \cite{robins1989analysis, robins1992g, vansteelandt2014structural, tian2014simple}. Regardless of the approach, $\widehat \psi_{\text{\tiny te}}(z, z^\prime,s^*)$ converges in probability to $\psi(z, z^\prime,s^*)$ when $\widehat b(z, z^\prime,s^*; X)$ is a consistent estimator of the conditional average treatment effect in trial $s^*$.

\vspace{0.1in}
\noindent
\emph{Estimation by modeling the probability of trial participation and treatment:} A second option is to use a weighting estimator based on (\ref{eq:identification_w_single}), 
\begin{equation}\label{eq:estimation_w_single}
  \widetilde \psi_{\text{\tiny w}}(z, z^\prime, s^*) = \Bigg\{\sum\limits_{i=1}^{n} I(S_i = 0) \Bigg\}^{-1} \sum\limits_{i=1}^{n} \left( \dfrac{I(Z_i = z)}{ \widehat \ell(z, s^*; X_i)} - \dfrac{I(Z_i = z^\prime)}{\widehat \ell(z^\prime, s^*; X_i)} \right) I(S_i = s^*)  \dfrac{ \widetilde p(X_i) }{1 - \widetilde p(X_i)} Y_i ,
\end{equation}
where $\widetilde p(X)$ is an estimator of $\Pr[S = 0 | X, I(S \in \{ 0,  s^*\}) = 1]$; $\widehat \ell(z,s^*; X)$ and $\widehat \ell(z^\prime,s^*; X)$ are estimators for $\Pr[Z = z | X , S = s^*]$ and $\Pr[Z = z^\prime | X , S = s^*]$, respectively.

\subsection*{Estimation for pooling the trials}

\vspace{0.1in}
\noindent
\emph{Estimation by modeling the conditional average treatment effect:} The first option is to use an estimator based on (\ref{eq:identification_g_collection}),
\begin{equation}\label{eq:estimation_g_collection}
  \widehat \phi_{\text{\tiny te}}(z, z^\prime) = \Bigg\{\sum\limits_{i=1}^n I(S_i = 0) \Bigg\}^{-1} \sum\limits_{i=1}^{n} I(S_i = 0)  \widehat t(z, z^\prime; X_i),
\end{equation}
where $\widehat t(z, z^\prime; X)$ is an estimator of the common (across trials) conditional mean difference $\tau(z, z^\prime; X)$. A general way to obtain $\widehat t(z, z^\prime; X)$ is to estimate the parameters of a regression model for the following conditional mean function $$\E \left[ \left( \dfrac{I(Z = z)}{\Pr[Z = z | X, S, I(S \in \mathcal S) = 1]} -  \dfrac{I(Z = z^\prime)}{\Pr[Z = z^\prime | X, S, I(S \in \mathcal S) = 1]} \right) Y \Big| X, I(S \in \mathcal S) = 1  \right],$$ where the probabilities in the denominators of the fractions inside the parenthesis are known (or can be estimated) \cite{robins1989analysis, robins1992g, vansteelandt2014structural, tian2014simple}. This approach does not require the treatment assignment mechanism to be the same across trials. The estimator $\widehat \phi_{\text{\tiny te}}(z, z^\prime)$ converges in probability to $\phi(z, z^\prime)$ when $\widehat t(z, z^\prime; X)$ is a consistent estimator of the (common across trials) conditional average treatment effect in the collection of trials $\mathcal S$.

\vspace{0.1in}
\noindent
\emph{Estimation by modeling the probability of trial participation and treatment:} A second option is to use a weighting estimator based on (\ref{eq:identification_w_collection}), 
\begin{equation}\label{eq:estimation_w_collection}
  \widehat \phi_{\text{\tiny w}}(z, z^\prime) = \Bigg\{\sum\limits_{i=1}^n I(S_i = 0) \Bigg\}^{-1} \sum\limits_{i=1}^{n} \left( \dfrac{I(Z_i = z)}{ \widehat e(z; X_i, S_i)} - \dfrac{I(Z_i = z^\prime)}{\widehat e(z^\prime; X_i, S_i)} \right) I(S_i \in \mathcal S)  \dfrac{ \widehat p(X_i) }{1 - \widehat p(X_i)} Y_i ,
\end{equation}
where $\widehat p(X)$ is an estimator of $\Pr[S = 0 | X]$, and $\widehat e(z; X, S)$ and $\widehat e(z^\prime; X, S)$ are estimators for $\Pr[Z = z | X , S,  I( S \in \mathcal S) = 1]$ and $\Pr[Z = z^\prime | X , S,  I( S \in \mathcal S) = 1]$, respectively. This weighted estimator converges in probability to $\phi(z, z^\prime)$ when $\widehat p(X)$ is a consistent estimator of $\Pr[S = 0 | X]$ and $\widehat e(z; X, S)$ is a consistent estimator of $\Pr[Z = z | X, S, I( S \in \mathcal S) = 1 ]$. Correct specification of a model for $\Pr[Z = z | X, S,  I( S \in \mathcal S) = 1 ]$ is straightforward even if trials have different treatment assignment mechanisms (but may prove more challenging in pooled analyses of observational studies).

\subsection*{Inference}

To construct Wald-style confidence intervals for $\widehat \psi_{\text{\tiny te}}(z, z^\prime, s^*)$, $\widehat  \psi_{\text{\tiny w}}(z, z^\prime,s^*)$, $\widehat  \phi_{\text{\tiny te}}(z, z^\prime)$, or $\widehat \phi_{\text{\tiny w}}(z, z^\prime)$, when using parametric models, we can obtain ``sandwich'' estimators of the sampling variance for each estimator we have described \cite{stefanski2002}. Alternatively, we can use the non-parametric bootstrap \cite{efron1994introduction}. Straightforward inference approaches are possible because under our sampling model the observations are assumed to be independent. We leave extensions of the sampling model to allow for dependence among the individuals sampled into each trial for future work.

\subsection*{Implementation of the estimators}

In eAppendix \ref{appendix:code} we provide a link to \texttt{R} code implementing the estimators described above using the \texttt{geex} package \cite{geex_manual}. Our implementation uses parametric working models: we use estimating equations for binary or multinomial logistic regression models to estimate conditional probabilities; and estimating equations for linear regression models to estimate conditional expectations. All estimating equations can be easily replaced by those appropriate for other generalized linear models \cite{mccullagh1989generalized}, as needed. Because we estimate the working model parameters jointly with the target parameters of all estimators, the sampling variances appropriately account for uncertainty \cite{stefanski2002, lunceford2004}. The code can be modified to address non-adherence (and incomplete follow-up) as needed.

\section*{Violations of conditional exchangeability in measure }

Conditional exchangeability in measure (for individual studies or for the collection $\mathcal S$, as applicable) will often not hold exactly in applications, when some modifiers of the treatment effect are not included in $X$. For example, in the case of transporting inferences from a single trial $s^* \in \mathcal S$, comparing treatments $z$ and $z^\prime$ in $\mathcal Z$, it is possible that for some $x$ with $f(x, S = 0) > 0$, we will have $$ \E[Y^z - Y^{z^\prime} | X = x, S = s^* ] \neq \E[Y^z - Y^{z^\prime} | X = x , S = 0 ],$$ and none of the methods described earlier in this paper will be able to provide valid inferences. Furthermore, substantive experts may disagree on the plausibility of conditional exchangeability in measure in any given application. When conditional exchangeability in measure is implausible or controversial, it is prudent to perform sensitivity analyses to examine how the study conclusions would change under violations of the condition with different magnitudes \cite{robins2000c}. 

\paragraph{When transporting from individual studies in the meta-analysis:} For individual studies in the collection $\mathcal S$, we can perform sensitivity analysis by modifying the methods described in reference \cite{dahabreh2019sensitivitybiascor}. 

\emph{Sensitivity analysis model:} A convenient way to parameterize the sensitivity analysis is to assume the following \emph{sensitivity analysis model}, conditional on baseline covariates:
$$\E[Y^z - Y^{z^\prime} | X, S = 0 ] =   \E[Y^z - Y^{z^\prime} | X, S = s^*]  + u(s^*; X).$$ Here, $u(s^*; X)$ is a user-specified, possibly study-specific and covariate-dependent, \emph{bias function} that expresses the degree of residual effect modification by trial participation, within levels of the measured covariates in $X$ (see \cite{dahabreh2019sensitivitybiascor} for a similar approach to sensitivity analysis when transporting inferences from a single trial and for suggestions for how to choose bias correction functions in applications).

For the choice of $u(s^*; X)$ function for which the sensitivity analysis model holds, by taking expectations, we obtain 
\begin{equation}\label{eq:sensitivity_analysis_model1}
    \begin{split}
    \E[Y^z - Y^{z^\prime} | S = 0 ] &= \E\big[ \E[Y^z - Y^{z^\prime} | X , S = s^*]  + u(s^*; X) \big| S = 0\big] \\
    &= \E\big[ \E[Y^z - Y^{z^\prime} | X , S = s^* ] \big| S = 0 \big ] + \E\big[ u(s^*; X) \big| S = 0\big].
    \end{split}
\end{equation}

Suppose that conditions \emph{A1}, \emph{A2}, \emph{A3}, and \emph{A5} hold, but condition \emph{A4} (conditional exchangeability in measure) may not hold. In such cases, using the result in equation  (\ref{eq:sensitivity_analysis_model1}), we obtain
\begin{equation}\label{eq:sensitivity_analysis_model2}
    \begin{split}
    \E[Y^z - Y^{z^\prime} | S = 0 ] = \psi(z, z^\prime, s^*) + \E\big[ u(s^*; X) \big| S = 0\big].
    \end{split}
\end{equation}

\emph{Estimation and inference for sensitivity analysis:} The result in equation (\ref{eq:sensitivity_analysis_model2}) suggest that sensitivity analysis estimators can be obtained by adding the term $$\left\{ \sum\limits_{i=1}^{n} I(S_i = 0) \right\}^{-1} \sum\limits_{i=1}^{n} I(S_i = 0)  u(s^*; X_i)$$ to the estimators of $ \psi(z, z^\prime, s^*)$ in the previous section. This ``sensitivity analysis term'' converges in probability to $\E\big[ u(s^*; X) \big| S = 0\big]$, under mild conditions. The function $u(s^*; X)$ is not identifiable from the data and thus sensitivity analysis should be performed over a range of different functions. Confidence intervals can again be obtained using either ``sandwich'' estimates of the variance or bootstrap methods. For example, the code we provide can be easily modified to obtain ``sandwich'' estimates of the variance simply by adding the sensitivity analysis term to the summands of the estimating equations.

\paragraph{Sensitivity analysis for transporting from multiple trials:} Sensitivity analysis for transporting inferences from multiple trials would follow the same principles as outlined for the case of individual studies. In the absence of conditional exchangeability in measure from one or more of the trials in the collection $\mathcal S$ trials (i.e., if assumption \emph{B4} does not hold), the observed data implications, and the chain of equalities in \eqref{eq:chain_observed_data_implications} in particular, would not be expected to hold. Furthermore, the choice of bias correction functions for each study $ u(s; X)$ for each trial $s \in \mathcal S$ would be more complicated because the functions would need to satisfy the following constraint (if all trials in the collection are transportable):
\small
\begin{equation*}
    \begin{split}
    &\E[Y | X, S= 1, Z = z] - \E[Y | X, S= 1, Z = z^\prime] + u(1; X) \\
    &\quad\quad\quad= \E[Y | X, S = 2, Z = z] - \E[Y | X, S= 2, Z = z^\prime] + u(2; X) \\
    &\quad\quad\quad= \ldots \\
    &\quad\quad\quad= \E[Y | X, S = m, Z = z] - \E[Y | X, S= m, Z = z^\prime] + u(m; X).
    \end{split}
\end{equation*} \normalsize
We leave the development of methods for harnessing expert knowledge to choose appropriate $u(s; X)$ functions for sensitivity analysis as future work.

\section*{Homogeneity of transported effects}

When investigators are willing to assume (and do not have contrary evidence) that exchangeability in measure holds for two or more trials in the collection $\mathcal S$, the observed data implications about conditional average treatment effects, have additional implications for the transported effects. Suppose that identifiability conditions \emph{B1} through \emph{B5} hold, then it has to be that \small
\begin{equation*}
    \begin{split}
&\E  \big[ \E[Y | X, S = 1, Z =z ]  - \E[Y | X, S = 1, Z =z^\prime ]  \big| S = 0 \big] \\
&\quad=\E \big[ \E[Y | X, S = 2, Z = z ]  - \E[Y | X, S = 2, Z = z^\prime ]  \big| S = 0 \big] \\
&\quad=\ldots \\
&\quad=\E \big[ \E[Y | X, S = m, Z = z ]  - \E[Y | X, S = m, Z = z^\prime ]  \big| S = 0 \big],
    \end{split}
\end{equation*}\normalsize
where, as defined earlier, $m$ is the total number of trials in $\mathcal S$. This implication of conditions \emph{B1} through \emph{B5} can be interpreted as a condition of homogeneity of transported inferences among the trials in $\mathcal S$. 

It is instructive to compare the homogeneity of transported inferences implication against the usual meta-analytic assumption of homogeneity of the study-specific average treatment effects in the same collection of trials,\small
\begin{equation*}
    \begin{split}
&\E  \big[ \E[Y | X, S = 1, Z =z ]  - \E[Y | X, S = 1, Z =z^\prime ]  \big| S = 1 \big] \\
&\quad=\E \big[ \E[Y | X, S = 2, Z = z ]  - \E[Y | X, S = 2, Z = z^\prime ]  \big| S = 2 \big] \\
&\quad=\ldots \\
&\quad=\E \big[ \E[Y | X, S = m, Z = z ]  - \E[Y | X, S = m, Z = z^\prime ]  \big| S = m \big].
    \end{split}
\end{equation*}\normalsize
This chain of equalities represents the null hypothesis adopted for most ``tests for heterogeneity'' used in applied meta-analyses \cite{laird1990some}. The critical difference is this: the homogeneity of transported inferences implication of conditions \emph{B1} through \emph{B5} is a comparison of functionals that involve marginalization over (i.e., standardization to) the same baseline covariate distribution, that of the target population $S=0$. In contrast, the usually homogeneity assumption in meta-analyses is a comparison of functionals that involve marginalization over a different baseline covariate distribution, that of each trial $S = s$. The latter may not be equal, even if assumptions \emph{B1} through \emph{B5} hold, when the distribution of $X$ varies across trials.

\section*{Application of the methods to data from the HALT-C trial}

\subsection*{Using data from a multicenter trial to emulate a meta-analysis}

The HALT-C trial enrolled 1050 patients with chronic hepatitis C and advanced fibrosis who had not responded to previous therapy and randomized them to treatment with peginterferon alfa-2a ($z=1$) versus no treatment ($z=0$). Patients were recrutied in 10 research centers and followed up every 3 months after randomization. We used data on the secondary outcome of platelet count at 9 months of follow-up; we report all outcome measurements as platelets $\times 10^3/$ml. For simplicity, we restricted our analysis to 974 patients with complete baseline covariate and outcome data. 

We used the HALT-C trial data to emulate a meta-analysis: First, we treated the data from the center that contributed the most patients as a sample from the target population ($S=0$; 202 patients). Second, we treated the data from the remaining 9 centers ($S = 1, \ldots, 9$; 772 patients) as if derived from a collection of separate randomized trials ($\mathcal S$). Third, we transported inferences from $\mathcal S$ to $S=0$ using the methods described in this paper. Fourth, we transported inferences from each trial $s$ in the collection $\mathcal S$ to $S=0$ using the methods described in \cite{dahabreh2019transportingStatMed}. 

The benefit of our approach for evaluating the methods is that by re-purposing data from a multi-center trial we have access to information on the randomly assigned treatment and outcome from $S=0$, allowing us to compare analyses using data exclusively from $S=0$ to the results of transportability analyses. These comparisons are informative because, provided the conditions needed for transporting inferences from the collection $\mathcal S$ (or from each center $s \in \mathcal S$) to $S = 0$ hold, we expect that estimates from transportability analyses should agree with the estimates from the analyses that exploit randomization in $S = 0$ (up to sampling variability). 

Our HALT-C analysis was not considered human subjects research because it used de-identified secondary data.

\subsection*{Methods implemented and comparisons}

We applied the estimators $\widehat \phi_{\text{\tiny te}}(1, 0)$ and $\widehat \phi_{\text{\tiny w}}(1, 0)$ to transport inferences from the collection of trials to $S = 0$. We also transported inferences separately from each trial to $S =0$ using the estimators $\widehat \psi_{\text{\tiny te}}(1, 0, s)$ and $\widehat \psi_{\text{\tiny w}}(1, 0, s)$ for each $s \in \mathcal S = \{1, \dots, 9\}$.

We specified parametric working models, as needed for each estimator (binary or multinomial logistic regression models for discrete outcomes, and linear regression models for continuous quantities).All working models used the following baseline covariates as main effects: baseline platelet count, age, sex, previous use of pegylated interferon, race, white blood cell count, history of injected recreational drugs, ever receiving a transfusion, body mass index, creatinine levels, and smoking status (to avoid numerical issues in the smallest centers, when transporting inferences from center 1 alone, we did not consider sex, previous use of pegylated interferon, and creatinine; when transporting inferences from center 9 alone, we did not consider sex, previous use of pegylated interferon, race, and creatinine).

\subsection*{Data analysis results}

The results from the meta-analysis emulation are summarized in Table \ref{tab:HALT_C_results} and graphed in the forest plot \cite{lewis2001forest} of Figure \ref{fig:forest_plot}. When transporting inferences from the collection $\mathcal S$ to $S = 0$, the treatment effect and weighting estimator results were -43.7 and -42.4, respectively (the confidence interval was slightly narrower for the treatment effect estimator). These transportability results agree with the randomization-based analysis from $S = 0$ (using treatment and outcome data only from the target center) that produced an estimate of -45.7. It is interesting to note that the confidence interval for the crude mean difference in $S = 0$ was substantially wider than the confidence intervals for the two transportability estimates from $\mathcal S$ to $S = 0$; this reflects the fact that the collection $\mathcal S$ contains many more individuals with covariate, treatment, and outcome data (772 individuals) compared to the $S = 0$ sample (202 individuals). In this particular dataset, the additional information is evidently enough to overcome any imprecision induced by differences in the covariate distribution between the trials in $\mathcal S$ and $S = 0$). 

In transportability analyses from each center $s \in \mathcal S$ to $S = 0$, point estimates from the treatment effect and weighting estimator were generally similar (and the former had narrower associated confidence intervals). For some centers, transportability analyses produced estimates that were reasonably close to the overall meta-analysis estimate, but the pattern was variable. For example, for $S = 1$, the transportability estimates (-50.9 for the treatment effect estimator; -48.9 for the weighting estimator) were closer to the overall transportability result (-43.7 for the treatment effect estimator; -42.4 for the weighting estimator) and the randomization-based analysis using $S = 0$ data (-45.7), compared to the analysis using only data from $S = 1$ (-29.3). Given the small sample size in some of the centers, analyses transporting inferences one trial at a time may be affected by sampling variability, over-fitting the data from individual centers, or model misspecification (because we relied on relatively simple parametric models).

\section*{Discussion}

In our experience, users of evidence syntheses are not interested in the aggregate of the populations underlying the completed randomized trials, which is often ill-defined because of convenience sampling. Instead, for decision-making, the users typically have a new target population in mind. Traditional approaches to meta-analysis, however, yield pooled effect estimates that are not generally interpretable as the causal effect in that target population of interest.  

Here, we described the conditions under which estimates from individual studies and pooled estimates can be causally interpreted as treatment effects in a target population which is chosen on substantive grounds and may be different from any of the of the populations sampled in the randomized trials. For example, our methods can be used when policy-makers want to examine the implications of a collection of trials for treatment effectiveness in a well-defined target population from which baseline covariate information is routinely collected (e.g., electronic health record data gathered by a healthcare system). The methods can also be used in planning a new trial when two or more related trials are already available: once the new trial's target population is specified, covariate data from it can be collected, and our methods can be used to obtain a treatment effect estimate in the target population to determine the feasibility of the new trial (e.g., through sample size calculations). In the rare case where one of the randomized trials in $\mathcal S$ is representative of the target population, the methods can be modified to treat the trial data as a sample from the target population.

The methods we propose relate to the general theory of causal identification when transporting inferences from multiple trials to a new setting \cite{bareinboim2013, bareinboim2013causal}. We address transportability from multiple trials by considering the underlying sampling model and provide both conditional treatment effect modeling and weighting approaches that can be used to reduce dependence on the specification of models for the conditional average treatment effect (or the conditional outcome mean) in the trials \cite{robins2001}.

Other work on the causal interpretation of meta-analyses has focused on identifiability conditions and their use to examine the presence of heterogeneity across studies \cite{sobel2017causal}; the causal interpretation of meta-analyses using published aggregate (summary) data \cite{kabali2016transportability}; or the use of aggregate data to estimate causal quantities in a ``meta-population'' that contains the individual superpopulations from each study included in a meta-analysis \cite{schnitzer2016causal}. The estimators we propose are different from those used to illustrate a recently proposed framework for causally interpretable meta-analysis \cite{sobel2017causal}; estimation in that work relied on conventional multivariable regression models for conditional average treatment effects. 

A recent unpublished technical report \cite{manski2019meta} addresses the same limitations of conventional meta-analysis methods as our work. Our approach is complementary to that proposed in the report: our identifiability conditions are stronger, but sufficient for point identification of causal effects in the target population (and, as noted earlier, sensitivity analysis can be done); the assumptions in the report appear weaker, but only allow partial (interval) identification. Our estimators allow statistical inference using standard methods (e.g., M-estimation or the bootstrap) and require access to individual-level data. In contrast, we expect that inference about the bounds in the report will often prove challenging because the target quantities are non-smooth and study-level data are limited (most meta-analyses include only a few studies, often with small sample sizes).

Our approach delineates the (non-parametric) identifiability conditions from any additional modeling assumptions that may be needed for estimation, especially when the vector of baseline covariates is high dimensional \cite{robins1997toward}. In this paper, we focused on simple conditional treatment effect and weighting estimators that can be easily implemented in standard statistical packages. A downside of these estimators is that, to produce valid results, the working models each estimator relies on need to be correctly specified. Also, it is possible to obtain doubly robust \cite{bang2005} estimators of the causal quantities of interest.

In summary, we have taken steps towards causally interpretable meta-analysis by considering the identification and estimation of average treatment effects when transporting inferences from a collection of randomized trials to a well-defined target population. To deploy the methods for applied evidence synthesis, extensions will be needed to address failure-time outcomes, systematically missing data, and measurement error in covariates or outcomes.

\clearpage 
\section*{Figure}

\begin{figure}[ht!]
	\centering
	\caption{Forest plot summarizing analyses using the multi-center HALT-C trial data to emulate a meta-analysis.}\label{fig:forest_plot}
	\includegraphics[width=12cm]{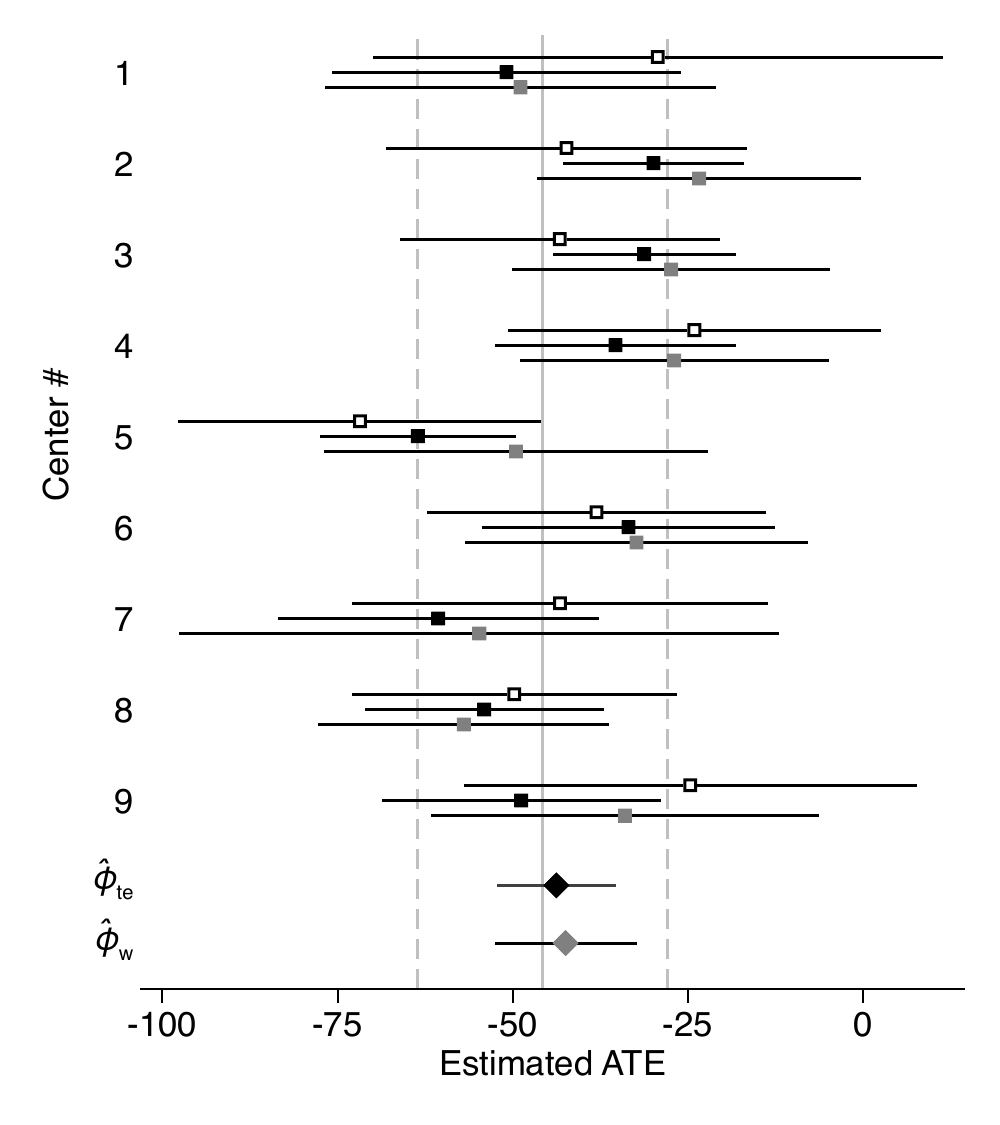}
  \caption*{Point estimates (markers) and 95\% confidence intervals (extending lines) from analyses using the multi-center HALT-C trial data to emulate a meta-analysis. White squares, $\square$ = unadjusted analyses that use data only from each center $s \in \mathcal S = \{ 1, \ldots, 9 \}$; black squares, $\blacksquare$ = transportability analyses using conditional outcome mean modeling and estimator \eqref{eq:estimation_g_single} to transport inferences from each center $S = s, s \in \{ 1, \ldots, 9 \}$ to $S = 0$; gray squares, $\graysquare$ = transportability analyses using inverse odds weighting and estimator \eqref{eq:estimation_w_single} to transport inferences from each center $S = s, s \in \{ 1, \ldots, 9 \}$ to $S = 0$; diamonds denote analyses transporting inferences from the collection of centers $\mathcal S = \{ 1, \ldots, 9 \}$ to S = 0: black diamond, $\blackdiamond$ = treatment effect estimator \eqref{eq:estimation_g_collection}, $\widehat\phi_{\text{\tiny te}}(1,0)$; gray diamond, $\graydiamond$ = weighting estimator \eqref{eq:estimation_w_collection}, $\widehat\phi_{\text{\tiny w}}(1,0)$. The solid vertical line indicates the point estimate and the dashed vertical lines indicate the limits of the confidence interval from the randomization-based analysis using only data from $S=0$.}
	\label{fig:densities}
\end{figure}

\clearpage
\section*{Table}

\begin{table}[ht!]
\centering
\caption{Results from analyses using the multi-center HALT-C trial data to emulate a meta-analysis.}\label{tab:HALT_C_results}
\label{Table_1_results}
\begin{tabular}{ccccc}
\toprule
$S$ & Sample size & Unadjusted & Treatment effect modeling & Weighting \\ \midrule
1 & 48 & -29.3 (-70.0, 11.5)  & -50.9 (-75.8, -25.9) & -48.9 (-76.8, -20.9)    \\
2 & 98 & -42.3 (-68.1, -16.4) & -29.9 (-42.8, -17.0) & -23.3 (-46.5, -0.2) \\
3 & 133 &  -43.3 (-66.1, -20.4) & -31.2 (-44.2, -18.2) & -27.4 (-50.0, -4.7) \\
4 & 69 & -24.0 (-50.6, 2.6)   & -35.3 (-52.5, -18.1) & -26.9 (-49.0, -4.8)  \\
5 & 77 & -71.8 (-97.7, -45.9) & -63.5 (-77.6, -49.5) & -49.5 (-76.9, -22.1)  \\
6 & 110 & -38.0 (-62.2, -13.8) & -33.4 (-54.3, -12.6) & -32.3 (-56.8, -7.8)  \\
7 & 94 & -43.2 (-72.9, -13.5) & -60.6 (-83.5, -37.7) & -54.7 (-97.6, -11.9)    \\
8 & 100 & -49.7 (-72.9, -26.6) & -54.0 (-71.1, -37.0) & -57.0 (-77.7, -36.3) \\
9 & 43 & -24.6 (-57.0, 7.7)   & -48.8 (-68.7, -28.8) & -33.9 (-61.7, -6.2)   \\ \midrule
0 & 202 & -45.7 (-63.6, -27.9) & -43.7 (-52.2, -35.2) & -42.4 (-52.6, -32.3) \\ \bottomrule
\end{tabular}
\caption*{Point estimates and 95\% confidence intervals (in parentheses) from analyses using the multi-center HALT-C trial data to emulate a meta-analysis. Unadjusted results are obtained using the crude mean difference in each trial for all rows. Treatment effect modeling results are obtained using the estimator of equation \eqref{eq:estimation_g_single} in rows 1 through 9 and estimator \eqref{eq:estimation_g_collection} in the bottom row; weighting results are obtained using the estimator of equation \eqref{eq:estimation_w_single} in rows 1 through 9 and estimator \eqref{eq:estimation_w_collection} in the bottom row. In rows 1 through 9, for each trial unadjusted results use only treatment and outcome data from each trial $s \in \mathcal S$; transportability analyses use covariate, treatment, and outcome data from each trial $s \in \mathcal S$ and baseline covariate data from $S = 0$. In the bottom row, the unadjusted analysis uses only treatment and outcome data from the center $S = 0$; transportability analyses use covariate, treatment, and outcome data from all trials in the collection $\mathcal S$ and baseline covariate data from $S = 0$.}
\end{table}

\clearpage
\bibliographystyle{unsrt}
\bibliography{Transporting_multiple_trials}


\clearpage
\setcounter{page}{1}
\appendix
\renewcommand{\appendixname}{eAppendix}

\section{Remarks on the sampling model}\label{appendix:sampling_model}
\renewcommand{\theequation}{A.\arabic{equation}}
\setcounter{equation}{0}

Combining data across multiple data sources requires some sort of sampling model that underlies the observed data and specifies the relationships between the data sources. For example, in ``conventional'' random effects meta-analysis models \cite{dersimonian1986meta}, the true study effects (e.g., the study-specific population-averaged treatment effects underlying each published study) are assumed to be independently sampled from an infinite population of study effects (e.g., \cite{dersimonian1986meta, higgins2009re}).

In this paper we operate under a sampling model that relates to the sampling model we proposed for generalizability and transportability analyses that involve a single randomized trial using non-nested trial designs \cite{dahabreh2019studydesigns}. Specifically, we assume that the research question defines an infinite superpopulation of individuals that can be stratified in several non-overlapping groups: (1) the non-overlapping groups of individuals who belong in the population underlying each trial $s$ in the collection $\mathcal S = \{1, \ldots, m\}$; that is, individuals who receive care or services in each trial's recruitment centers, who meet that trial's eligibility criteria, and who would be invited and agree to participate in the trial; and (2) individuals who belong to the target population of substantive interest, $S = 0$, and who are not invited to participate in any trial or who would decline if invited. Individuals with $S = 0$ may also be a well-defined subset of individuals who were not invited to participate in any trial or who were invited to participate but declined; that is, individuals with $S = 0$ do not need to exhaust the set of individuals not participating in any trial. 

The relative size of subsets of the superpopulation defined by different $S$ values is determined by the scientific question, as well as the trial eligibility criteria, individuals' access to centers recruiting patients into different trials, and individuals' preferences (e.g., desire to participate in a trial). This infinite superpopulation \cite{robins1988confidence} is a convenient fiction \cite{hernan2020}, but it seems to us to be somewhat more plausible than the infinite population of study effects invoked by random effects meta-analyses.  

We assume that the data are generated by random sampling of individuals from the superpopulation, stratified by $S$, with sampling fractions are (1) constant for individuals from the same subset of the superpopulation, $S = s$, for $s\in\{0, 1, \ldots,m\}$; (2) possibly variable across groups defined by different $S$ values; (3) and unknown to the investigators. Constancy of the sampling fraction for individuals with $S = s$ for every $s \in \{0, 1 , \ldots, m\}$ is a reflection of the frequently entertained assumption that individuals participating in a trial can be viewed as a simple random sample from the population underlying the trial; as well at the assumption that the sample of the target population is representative of that population. Variation of the sampling fractions across $S$ values reflects differences in recruitment practices across trials as well as the inability to obtain data from all individuals from the target population. The sampling fractions are unknown to the investigators because they depend on complex mechanisms that drive the design and conduct of randomized trials, as well as the underlying population of each trial. 

The sampling step induces a biased sampling model \cite{bickel1993efficient} in the following sense: in the data, the ratios $\dfrac{n_j}{ n} = \dfrac{n_j}{ n_0 + \sum_{s \in \mathcal S} n_s}$, for $j \in \{0,1, \ldots, m\}$, do not reflect the population probabilities of belonging to the subset of the population with $S = j$, because the sampling fraction from each subset is unknown (and possibly variable between subsets). As a technical condition, we require that as $n \longrightarrow \infty$, $\dfrac{n_j}{ n} \longrightarrow \pi_j > 0$, for $j \in \{0,1, \ldots, m\}.$ Nevertheless, under the biased sampling model, the limiting values, $\pi_j$, are not equal to the superpopulation probabilities $\Pr[S = j]$.

In effect, we view the data as generated by stratified random sampling of the actual population with sampling fractions that only depend on $S$ but are unknown to the investigators. Arguments analogous to those in \cite{dahabreh2019transportingStatMed} and \cite{dahabreh2019studydesigns} show that the identifiability of the causal quantities of interest is unaffected by biased sampling.

\clearpage
\section{Identification of the effect of treatment assignment}\label{appendix:identification}
\renewcommand{\theequation}{B.\arabic{equation}}
\setcounter{equation}{0}

In this eAppendix, we collect results about the identification of average treatment effects of treatment assignment (intention-to-treat average treatment effects) in the target population.

\noindent
\textbf{Transporting from an individual study}

\begin{proposition}\label{prop:gform_SZ}
Under conditions A1 through A5 in the main text, the average treatment effect in the target population, comparing treatments $z$ and $z^\prime$ in $\mathcal Z$, $\E[Y^z - Y^{z^\prime} | S = 0]$, using data from study $s^* \in \mathcal S$ and the target population, is identified by the following functional of the observed data distribution:
\begin{equation} \label{eq:appendix_identification_g_single}
 \psi(z, z^{\prime},s^*) \equiv \E \big[  \E[Y | X, S = s^*, Z = z ] - \E[Y | X, S = s^*, Z = z^\prime ]  \big| S = 0 \big].
\end{equation}
\end{proposition}

\noindent
\begin{proof} 
\begin{equation*} 
  \begin{split}
  &\E[Y^z - Y^{z^\prime} | S = 0] \\
    &\quad\quad= \E\big[ \E[Y^z - Y^{z^\prime} | X, S = 0] \big| S = 0 \big] \\
    &\quad\quad= \E\big[ \E[Y^z - Y^{z^\prime} | X, S = s^* ] \big| S = 0 \big] \\
    &\quad\quad= \E\big[ \E[Y^z | X, S = s^* ] - \E[ Y^{z^\prime} | X, S = s^* ] \big| S = 0 \big] \\
    &\quad\quad= \E\big[ \E[Y^z | X, S = s^*, Z = z ] - \E[ Y^{z^\prime} | X, S = s^* , Z = z^\prime] \big| S = 0 \big] \\
    &\quad\quad= \E\big[ \E[Y | X, S = s^*, Z = z ] - \E[ Y| X, S = s^* , Z = z^\prime] \big| S = 0 \big] \\
    &\quad\quad  \equiv \psi(z, z^{\prime},s^*).
  \end{split}
\end{equation*} 
\end{proof}

\clearpage

\begin{proposition}\label{prop:appendix_identification_w_single}
Under positivity conditions A3 and A5, \footnotesize
\begin{equation}\label{eq:appendix_identification_w_single}
    \begin{split}
  \psi(z, z^{\prime}, s^*) &= \dfrac{1}{\Pr[S = 0 ]}  \E\Bigg[  \left( \dfrac{I(Z = z )}{\Pr[Z = z | X, S = s^*]} - \dfrac{I(Z = z^\prime )}{\Pr[Z = z^\prime | X,S  = s^*]}  \right)  \dfrac{ I(S=s^*) Y \Pr[S = 0| X, I(S \in \{0, s^*\})= 1] }{ \Pr[ S = s^*| X, I(S \in \{0, s^*\})= 1]  } \Bigg].
  \end{split}
\end{equation}
\end{proposition} \normalsize

\begin{proof}
We begin from the left-hand-side, using the definition of $\psi(z, z^\prime, s^*)$, \footnotesize
\begin{equation*}
  \begin{split}
\psi(z, z^{\prime}, s^*) &\equiv \E\big[ \E[Y | X, S = s^*, Z = z ] - \E[ Y| X, S = s^* , Z = z^\prime] \big| S = 0 \big] \\
    &= \E\Bigg[ \E\left[ \dfrac{ I(Z = z) Y}{\Pr[Z = z | X, S = s^*]}  \Big| X, S = s^* \right] - \E\left[ \dfrac{I(Z = z^\prime ) Y}{\Pr[Z = z^\prime | X, S = s^*]} \Big| X, S = s^* \right] \Bigg| S = 0 \Bigg] \\
    &= \E\Bigg[ \E\left[ \left( \dfrac{ I(Z = z) }{\Pr[Z = z | X, S = s^*]}  - \dfrac{I(Z = z^\prime )}{\Pr[Z = z^\prime | X, S = s^*]} \right) Y \Big| X, S = s^* \right] \Bigg| S = 0 \Bigg] \\
    &= \E\Bigg[ \E\left[ \left( \dfrac{ I(Z = z) }{\Pr[Z = z | X, S = s^*]}  - \dfrac{I(Z = z^\prime )}{\Pr[Z = z^\prime | X, S = s^*]} \right) \dfrac{I(S = s^*)Y}{\Pr[S = s^* | X]} \Big| X \right] \Bigg| S = 0 \Bigg] \\ 
    &= \dfrac{1}{\Pr[S = 0]} \E\Bigg[ I(S = 0) \E\left[ \left( \dfrac{ I(Z = z) }{\Pr[Z = z | X, S = s^*]}  - \dfrac{I(Z = z^\prime )}{\Pr[Z = z^\prime | X, S = s^*]} \right) \dfrac{I(S = s^*)Y}{\Pr[S = s^* | X]} \Big| X \right] \Bigg] \\ 
    &= \dfrac{1}{\Pr[S = 0]} \E\Bigg[ \Pr[S = 0 | X] \E\left[ \left( \dfrac{ I(Z = z) }{\Pr[Z = z | X, S = s^*]}  - \dfrac{I(Z = z^\prime )}{\Pr[Z = z^\prime | X, S = s^*]} \right) \dfrac{I(S = s^*)Y}{\Pr[S = s^* | X]} \Big| X \right] \Bigg] \\
    &= \dfrac{1}{\Pr[S = 0]} \E\Bigg[  \E\left[ \left( \dfrac{ I(Z = z) }{\Pr[Z = z | X, S = s^*]}  - \dfrac{I(Z = z^\prime )}{\Pr[Z = z^\prime | X, S = s^*]} \right) \dfrac{I(S = s^*)Y \Pr[S = 0 | X] }{\Pr[S = s^* | X]} \Big| X \right] \Bigg] \\
    &= \dfrac{1}{\Pr[S = 0]} \E\Bigg[   \left( \dfrac{ I(Z = z) }{\Pr[Z = z | X, S = s^*]}  - \dfrac{I(Z = z^\prime )}{\Pr[Z = z^\prime | X, S = s^*]} \right) \dfrac{I(S = s^*)Y \Pr[S = 0 | X] }{\Pr[S = s^* | X]}  \Bigg] \\
    &= \dfrac{1}{\Pr[S = 0]} \E\Bigg[   \left( \dfrac{ I(Z = z) }{\Pr[Z = z | X, S = s^*]}  - \dfrac{I(Z = z^\prime )}{\Pr[Z = z^\prime | X, S = s^*]} \right) \dfrac{I(S = s^*)Y \Pr[S = 0 | X, I(S \in \{0, s^*\})= 1] }{\Pr[S = s^* | X, I(S \in \{0, s^*\})= 1]}  \Bigg],
  \end{split}
\end{equation*} \normalsize
where the last step uses the fact that $\dfrac{ \Pr[S = 0 | X] }{\Pr[S = s^* | X]} = \dfrac{ \Pr[S = 0 | X, I(S \in \{0, s^*\})= 1] }{\Pr[S = s^* | X, I(S \in \{0, s^*\})= 1]} $.

\end{proof}

\begin{remark}
The derivation for Proposition \ref{prop:appendix_identification_w_single} does not use any conditions that involve potential outcomes and thus holds whether or not $\psi(z, z^\prime, s^*)$ has a causal interpretation.
\end{remark}

\clearpage
\noindent
\textbf{Transporting the entire collection of trials}

\begin{proposition}\label{prop:appendix_identification_g_collection}
If conditions B1 through B5 from the main text, hold for every trial $s \in \mathcal S$, then the average treatment effect in the target population, comparing treatments $z$ and $z^\prime$ in $ \mathcal Z$, $\E[Y^z - Y^{z^\prime} | S = 0]$, equals the following functional of the observed data distribution:
\begin{equation} \label{eq:appendix_identification_g_collection}
 \phi(z, z^{\prime}) \equiv \E \big[  \tau (z, z^\prime; X)  \big| S = 0 \big],
\end{equation}
where $\tau (z, z^\prime; X)$ is the common (across trials) observed outcome mean difference comparing treatments $z$ and $z^\prime$ conditional on $X$.
\end{proposition}

\begin{proof}
If conditions \emph{B4} and \emph{B5}, hold for every trial $s \in \mathcal S = \{ 1, \ldots, m \}$ then, as shown in the main text, 
\begin{equation*}
    \E[Y^z - Y^{z^\prime}| X, S= 1] = \ldots = \E[Y^z - Y^{z^\prime} | X, S = m] = \E[Y^z - Y^{z^\prime} | X, I(S \in \mathcal S) = 0].
\end{equation*}   

Under conditions 1 through 3, the above result implies, 
\begin{equation*}
    \E[Y | X, S= 1, Z = z] - \E[Y | X, S= 1, Z = z^\prime] = \ldots = \E[Y | X, S = m, Z = z] - \E[Y | X, S= m, Z = z^\prime].
\end{equation*}

Using the notation, $\tau (z, z^\prime; X)$ to denote the common (across trials) observed outcome mean difference comparing treatments $z$ and $z^\prime$ conditional on $X$, and Proposition \ref{prop:gform_SZ}, we obtain 
\begin{equation}
 \phi(z, z^{\prime}) = \E \big[  \tau (z, z^\prime; X)  \big| S = 0 \big].
\end{equation}
\end{proof}

\clearpage
\begin{proposition}\label{prop:appendix_identification_g_SZ_collection}
Under positivity conditions B3 and B5, 
\begin{equation}\label{eq:appendix_identification_w_collection}
    \begin{split}
  \phi(z, z^{\prime}) &= \dfrac{1}{\Pr[S = 0 ]}  \E\Bigg[  \omega(z, z^\prime; X, S)  \dfrac{ Y \Pr[S = 0| X] }{\Pr[ I(S \in \mathcal S) = 1| X]  } \Bigg],
  \end{split}
\end{equation}
where $$\omega(z, z^\prime; X, S) =  \left( \dfrac{I(Z = z )}{\Pr[Z = z | X, S, I(S \in \mathcal S) = 1]} - \dfrac{I(Z = z^\prime )}{\Pr[Z = z^\prime | X,S, I(S \in \mathcal S) = 1]}  \right) I(S \in \mathcal S).$$
\end{proposition}

\begin{proof}  
First, it is easy to see that for each trials $s \in \mathcal S$, 
\begin{equation}\label{eq:appendix_chain}
    \begin{split}
        &\E[Y | X, S = s, Z = z] - \E[Y | X, S = s, Z = z^\prime] \\
        &\quad\quad\quad\quad\quad\quad= \E \left[ \left( \dfrac{I(Z = z)}{\Pr[Z = z | X, S = s]} - \dfrac{I(Z = z^\prime)}{\Pr[Z = z^\prime | X, S = s]} \right) Y \Big| X, S = s \right] \\ 
        &\quad\quad\quad\quad\quad\quad= \E \left[ \left( \dfrac{I(Z = z)}{\Pr[Z = z | X, S, I(S \in \mathcal S) = 1]} - \dfrac{I(Z = z^\prime)}{\Pr[Z = z^\prime | X, S, I(S \in \mathcal S) = 1]} \right) Y \Big| X, S = s \right]. 
    \end{split}
\end{equation}

From the derivation of Proposition \ref{prop:appendix_identification_g_collection},
\begin{equation*}
    \E[Y | X, S= 1, Z = z] - \E[Y | X, S= 1, Z = z^\prime] = \ldots = \E[Y | X, S = m, Z = z] - \E[Y | X, S= m, Z = z^\prime],
\end{equation*}
which, combined with the result in \eqref{eq:appendix_chain}, implies that
\begin{equation}
    \begin{split}
        &\E \left[ \left( \dfrac{I(Z = z)}{\Pr[Z = z | X, S, I(S \in \mathcal S) = 1]} - \dfrac{I(Z = z^\prime)}{\Pr[Z = z^\prime | X, S, I(S \in \mathcal S) = 1]} \right) Y \Big| X, S = 1 \right] = \ldots \\
        &\quad\quad\quad\quad\quad\quad= \E \left[ \left( \dfrac{I(Z = z)}{\Pr[Z = z | X, S, I(S \in \mathcal S) = 1]} - \dfrac{I(Z = z^\prime)}{\Pr[Z = z^\prime | X, S, I(S \in \mathcal S) = 1]} \right) Y \Big| X, S = m \right].
    \end{split}
\end{equation}

Let $$\E \left[ \left( \dfrac{I(Z = z)}{\Pr[Z = z | X, S, I(S \in \mathcal S) = 1]} - \dfrac{I(Z = z^\prime)}{\Pr[Z = z^\prime | X, S, I(S \in \mathcal S) = 1]} \right) Y \Big| X, I(S \in \mathcal S) = 1 \right]$$ denote the quantity equal to all terms in the above chain of equations. Because this quantity is equal to each of the conditional mean differences $ \E[Y | X, S= s, Z = z] - \E[Y | X, S= s, Z = z^\prime]$, for every $s \in \mathcal S$, it has to be that $$ \tau(z, z^\prime; X) = \E \left[ \left( \dfrac{I(Z = z)}{\Pr[Z = z | X, S, I(S \in \mathcal S) = 1]} - \dfrac{I(Z = z^\prime)}{\Pr[Z = z^\prime | X, S, I(S \in \mathcal S) = 1]} \right) Y \Big| X, I(S \in \mathcal S) = 1 \right].$$

Using the derivation of Proposition \ref{prop:appendix_identification_g_collection}, we obtain
\begin{equation*}
    \begin{split}
    \phi(z, z^\prime) &= \E[\tau(z, z^\prime; X) | S = 0] \\
    &= \E \left[ \E \left[ \left( \dfrac{I(Z = z)}{\Pr[Z = z | X, S, I(S \in \mathcal S) = 1]} - \dfrac{I(Z = z^\prime)}{\Pr[Z = z^\prime | X, S, I(S \in \mathcal S) = 1]} \right) Y \Big| X, I(S \in \mathcal S) = 1 \right]  \Bigg| S = 0  \right].
    \end{split}
\end{equation*}

Because 
\footnotesize
\begin{equation*}
    \begin{split}
    &\E \left[ \E \left[ \left( \dfrac{I(Z = z)}{\Pr[Z = z | X, S, I(S \in \mathcal S) = 1]} - \dfrac{I(Z = z^\prime)}{\Pr[Z = z^\prime | X, S, I(S \in \mathcal S) = 1]} \right) Y \Big| X, I(S \in \mathcal S) = 1 \right]  \Bigg| S = 0  \right] \\
    &\quad= \dfrac{1}{\Pr[S = 0]} \E \Bigg[ I(S = 0) \E \left[ \left( \dfrac{I(Z = z)}{\Pr[Z = z | X, S, I(S \in \mathcal S) = 1]} - \dfrac{I(Z = z^\prime)}{\Pr[Z = z^\prime | X, S, I(S \in \mathcal S) = 1]} \right) Y \Big| X, I(S \in \mathcal S) = 1 \right]    \Bigg] \\
    &\quad= \dfrac{1}{\Pr[S = 0]} \E \Bigg[ \Pr[S = 0|X] \E \left[ \left( \dfrac{I(Z = z)}{\Pr[Z = z | X, S, I(S \in \mathcal S) = 1]} - \dfrac{I(Z = z^\prime)}{\Pr[Z = z^\prime | X, S, I(S \in \mathcal S) = 1]} \right) Y \Big| X, I(S \in \mathcal S) = 1 \right]    \Bigg] \\
     &\quad= \dfrac{1}{\Pr[S = 0]} \E \Bigg[ \E \left[ \left( \dfrac{I(Z = z)}{\Pr[Z = z | X, S, I(S \in \mathcal S) = 1]} - \dfrac{I(Z = z^\prime)}{\Pr[Z = z^\prime | X, S, I(S \in \mathcal S) = 1]} \right) \dfrac{I(S \in \mathcal S)Y \Pr[S = 0 | X]}{\Pr[I(S \in \mathcal S) = 1 | X]} \Big| X \right]    \Bigg] \\
     &\quad= \dfrac{1}{\Pr[S = 0]} \E \left[ \left( \dfrac{I(Z = z)}{\Pr[Z = z | X, S, I(S \in \mathcal S) = 1]} - \dfrac{I(Z = z^\prime)}{\Pr[Z = z^\prime | X, S, I(S \in \mathcal S) = 1]} \right) \dfrac{I(S \in \mathcal S)Y \Pr[S = 0 | X]}{\Pr[I(S \in \mathcal S) = 1 | X]}     \right],
    \end{split}
\end{equation*} 
\normalsize we conclude that \small
$$\phi(z, z^\prime) = \dfrac{1}{\Pr[S = 0]} \E \left[ \left( \dfrac{I(Z = z)}{\Pr[Z = z | X, S, I(S \in \mathcal S) = 1]} - \dfrac{I(Z = z^\prime)}{\Pr[Z = z^\prime | X, S, I(S \in \mathcal S) = 1]} \right) \dfrac{I(S \in \mathcal S)Y \Pr[S = 0 | X]}{\Pr[I(S \in \mathcal S) = 1 | X]}     \right], $$ \normalsize
which completes the proof.
\end{proof}

\noindent
\begin{remark}
The derivation for Proposition \ref{prop:appendix_identification_g_SZ_collection} does not use any conditions that involve potential outcomes and thus holds whether or not $\phi(z, z^\prime)$ has a causal interpretation.
\end{remark}

\clearpage
\section{Addressing non-adherence in the randomized trials}\label{appendix:non_adherence}
\renewcommand{\theequation}{C.\arabic{equation}}
\setcounter{equation}{0}

In this section we sketch an approach to transporting per-protocol effects from a collection of randomized trials to a target population. We only consider the simple case of a binary adherence indicator (all-or-nothing adherence) and an outcome measured at the end of the study, because the strategy for combining adjustments for non-adherence and transportability in this simple case readily extends to more complicated non-adherence patterns (multiple time periods, more than two treatment recommendations, etc.), using well-known methods for the analysis of the randomized trials with non-adherence \cite{robins2000a, robins2000b}.

\subsection{Setup and notation}

We begin by introducing some notation to represent non-adherence and define per-protocol effects. In each trial $S$, information is collected on baseline covariates $X$, treatment assignment $Z$, post-treatment assignment covariates $L$, treatment received $A$, and the outcome $Y$. As in the main text, lower case letters denote realizations of the corresponding random variables. The set of possible assigned treatments is $\mathcal Z$; the set of possible received treatments is $\mathcal A$. Each specific pair of assigned and received treatment is denoted as $(z, a)$, $z \in \mathcal Z$, $a \in \mathcal A$. 

We are now interested in counterfactual outcomes under joint intervention to assign treatment $z$ and enforce the receipt of treatment $a$; we denote these counterfactual outcomes as $Y^{z,a}$. The average treatment effect comparing two different joint interventions, $(z, a)$ and $(z^\prime, a^\prime)$, in the target population is $\E[Y^{z,a} - Y^{z^\prime, a^\prime} | S = 0 ]$. For example, suppose that we are evaluating a binary treatment; then, in our setup, the most interesting causal contrast is arguably the per-protocol effect in the target population, $\E[Y^{z=1,a=1} - Y^{z^\prime=0,a^\prime=0} | S = 0]$.

\subsection{Identifiability conditions for transporting inferences from a single trial}

We assume that the following identifiability conditions hold for some trial $s^* \in \mathcal S$:

\vspace{0.1in}
\noindent
\emph{C1. Consistency of potential outcomes:} For every individual $i$ in the trials or target population, for every $z \in \mathcal Z$ and every $a \in \mathcal A$, if $Z_i = z$ and $A_i = a$, then $Y_i^{z,a} = Y_i$.

\vspace{0.1in}
\noindent
\emph{C2. Sequential conditional exchangeability for assignment and receipt of treatment in the trial:} for trial $s^* \in \mathcal S$, each treatment assignment $z \in \mathcal Z$, and each treatment received $a \in \mathcal A$, we have that $\E[Y^{z, a} | X = x, S = s^*, Z = z ] = \E[Y^{z, a} | X = x, S = s^*] $ for each $x$ with $f(x, S = s^*) > 0$; and $\E[Y^{z, a} | X = x, S = s^*, Z = z, L = l, A = a ] = \E[Y^{z, a} | X = x, S = s^*,  Z = z, L = l] $ for each $x$ and $l$ with $f(x, S = s^*,  Z = z, l) > 0$. 

\vspace{0.1in}
\noindent
\emph{C3. Sequential positivity of the probability of treatment assignment and receipt of treatment in the randomized trial:} for trial $s^* \in \mathcal S$ we have that $\Pr[Z = z | X = x, S = s^*] >0$ for every $x$ with $f(x, S = s^*) >0$;  and for each treatment $a \in \mathcal A$, $\Pr[A = a | X = x, S = s^*, Z= z, L = l] >0$ for every $x$ and $l$ with $f(x, S = s^*, Z = z, L = l) >0$. 

\vspace{0.1in}
\noindent
\emph{C4. Conditional exchangeability in measure for the per-protocol effect:} for trial $s^*$, for each pair of treatment assignments $z \in \mathcal Z$ and $z^\prime \in \mathcal Z$, and each pair of treatments received $a \in \mathcal A$ and $a^\prime \in \mathcal A$, $$\E[Y^{z,a} - Y^{z^\prime, a^\prime} | X = x, S = 0 ] = \E[Y^{z,a} - Y^{z^\prime, a^\prime} | X = x, S = s^*] $$ for every $x$ with $f(x, S =0)>0$.

\vspace{0.1in}
\noindent
\emph{C5. Positivity of the probability of participation in the trial:} $\Pr[S = s^* |X = x] > 0$ for every $x$ with $f(x, S = 0) > 0$. 

Conditions \emph{C1} through \emph{C3} are ``standard'' conditions used for the analysis of randomized trials with non-adherence \cite{hernan2020}. Condition \emph{C4} is a per-protocol effect version of condition \emph{A4} for intention-to-treat effects given in the main text; and condition \emph{C5} is the same as condition \emph{A5} given in the main text.

\subsection{Identification of per-protocol effects in the target population}

We will now show that under conditions \emph{C1} through \emph{C5} per-protocol effects in the target population, $\E[Y^{z,a} - Y^{z^\prime, a^\prime} | S = 0 ]$, are identifiable using randomized trial data on $(X, S = s^*, Z, L, A, Y)$ from the trial $s^*$, and data on $(X, S = 0)$ from the sample of the target population. 
 
Using conditions \emph{C4} and \emph{C5}, and taking expectations over the conditional distribution of $X$ in the target population, we have 
\begin{equation}\label{eq:identification_per_protocol1}
    \begin{split}
        &\E[Y^{z,a} - Y^{z^\prime, a^\prime} | S = 0 ] \\
        &\quad\quad= \E \big[ \E[Y^{z,a} - Y^{z^\prime, a^\prime} | X, S = 0] \big| S = 0 \big] \\
        &\quad\quad= \E \big[ \E[Y^{z,a} - Y^{z^\prime, a^\prime} | X, S = s^*] \big| S = 0 \big] \\
         &\quad\quad= \E \big[ \E[Y^{z,a}  | X, S = s^*] - \E[ Y^{z^\prime, a^\prime} | X, S = s^*]  \big| S = 0 \big].
    \end{split}
\end{equation}
To complete the identification analysis, note that the only terms involving counterfactual quantities in the last expression of display (\ref{eq:identification_per_protocol1}) are the trial-specific conditional potential outcome means $\E[ Y^{z, a} | X, S = s^*]$, for every $z \in \mathcal Z$ and every $a \in \mathcal A$. Using conditions \emph{C1} through \emph{C3}, we have 
\begin{equation}\label{eq:identification_per_protocol2}
    \begin{split}
\E[Y^{z,a} | X, S = s^*] &= \E[Y^{z,a} | X, S = s^*, Z = z] \\
     &= \E\big[  \E[Y^{z,a} | X, S = s^*, Z = z, L] \big| X, S = s^*, Z = z \big] \\
     &= \E\big[  \E[Y^{z,a} | X, S = s^*, Z = z, L, A = a] \big| X, S = s^*, Z = z \big] \\
     &= \E\big[  \E[Y | X, S = s^*, Z = z, L, A = a] \big| X, S = s^*, Z = z \big].
    \end{split}
\end{equation}

Define $ \theta(z, a, s^*; X) \equiv  \E\big[  \E[Y | X, S = s^*, Z = z, L, A = a] \big| X, S = s^*, Z = z \big]$ for each $z \in \mathcal Z$ and $a \in \mathcal A$. Combining the results from displays (\ref{eq:identification_per_protocol1}) and (\ref{eq:identification_per_protocol2}), we have
\begin{equation*}
    \begin{split}
&\E[Y^{z,a} - Y^{z^\prime, a^\prime} | S = 0 ] = \E \Big[  \theta(z, a, s^*; X) -  \theta(z^\prime, a^\prime, s^*; X)  \Big| S = 0 \Big],
    \end{split}
\end{equation*} 
which establishes the identifiability of per-protocol effects when transporting inferences from a randomized trial to a target population.

\subsection{Transporting inferences from a collection of trials}

Suppose that the identifiability conditions \emph{C1} through \emph{C5} hold for every trial $s \in \mathcal S$. Then, we would have 
\begin{equation}
    \begin{split}
        &\E[Y^{z,a} - Y^{z^\prime, a^\prime} | X, S = 1 ] = \ldots = \E[Y^{z,a} - Y^{z^\prime, a^\prime} | X, S = m] = \E[Y^{z,a} - Y^{z^\prime, a^\prime} | X, S = 0]. 
    \end{split}
\end{equation}

Using the results from the previous section, we obtain
\begin{equation}
    \begin{split}
        & \theta(z, a, 1; X) -  \theta(z^\prime, a^\prime, 1; X) = \ldots =  \theta(z, a, m; X) -  \theta(z^\prime, a^\prime, m; X) = \E[Y^{z,a} - Y^{z^\prime, a^\prime} | X, S = 0].
    \end{split}
\end{equation}
Note that, similar to our results in the main text, the chain of equalities
\begin{equation}
    \begin{split}
        \theta(z, a, 1; X) -  \theta(z^\prime, a^\prime, 1; X) = \ldots =  \theta(z, a, m; X) -  \theta(z^\prime, a^\prime, m; X) \equiv \lambda(z, a, z^\prime, a^\prime; X),
    \end{split}
\end{equation}
is an observed data implication of assuming transportability of the per-protocol effects across the collection $\mathcal S$, and we use the notation $\lambda(z, a, z^\prime, a^\prime; X)$ for that common (across trials) quantity.

Using the law of total expectation, and the identification result for a single trial, we obtain the following identification result for transporting the entire collection of trials $$\E[Y^{z,a} - Y^{z^\prime, a^\prime} | S = 0] = \E[ \lambda(z, a, z^\prime, a^\prime; X) | S = 0 ]. $$

Obtaining a weighting re-epression of these results is also instructive. For the single trial $s^* \in \mathcal S$, and every $z \in \mathcal Z$
\begin{equation}
    \begin{split}
&\theta(z, a, s^*; X) \\
&\quad \equiv  \E\big[  \E[Y | X, S = s^*, Z = z, L, A = a] \big| X, S = s^*, Z = z \big] \\
&\quad \E\left [ \dfrac{I(Z = z, A = a)}{\Pr[Z = z | X, S = s^*]\Pr[A = a| X, S = s^*, Z = z, L ]} Y \Big | X , S = s^* \right] \\
&\quad \E\left [ \dfrac{I(Z = z, A = a)}{\Pr[Z = z | X, S, I(S \in \mathcal S) = 1]\Pr[A = a| X, S, Z = z, L, I(S \in \mathcal S) = 1 ]} Y \Big | X , S = s^*, I(S \in \mathcal S) = 1 \right],
    \end{split}
\end{equation}
and consequently, 
\begin{equation}
    \begin{split}
&\theta(z, a, s^*; X) - \theta(z^\prime, a^\prime, s^*; X) \\
&\quad=   \E\Bigg [ \Bigg(  \dfrac{I(Z = z, A = a)}{\Pr[Z = z | X, S, I(S \in \mathcal S) = 1]\Pr[A = a| X, S, Z = z, L, I(S \in \mathcal S) = 1 ]} \\
&\quad  - \dfrac{I(Z = z^\prime, A = a^\prime)}{\Pr[Z = z^\prime | X, S, I(S \in \mathcal S) = 1]\Pr[A = a^\prime| X, S, Z = z^\prime, L, I(S \in \mathcal S) = 1 ]} \bigg) Y \Bigg | X , S = s^*, I(S \in \mathcal S) = 1 \Bigg].
    \end{split}
\end{equation}

It follows that, if all trials $s \in \mathcal S$ are transportable to the target population, then \small
\begin{equation*}
    \begin{split}
        &\lambda(z, a, z^\prime, a^\prime; X) =   \E\Bigg [ \Bigg(  \dfrac{I(Z = z, A = a)}{\Pr[Z = z | X, S, I(S \in \mathcal S) = 1]\Pr[A = a| X, S, Z = z, L, I(S \in \mathcal S) = 1 ]} \\
&\quad  - \dfrac{I(Z = z^\prime, A = a^\prime)}{\Pr[Z = z^\prime | X, S, I(S \in \mathcal S) = 1]\Pr[A = a^\prime| X, S, Z = z^\prime, L, I(S \in \mathcal S) = 1 ]} \bigg) Y \Bigg | X , I(S \in \mathcal S) = 1 \Bigg]. 
    \end{split}
\end{equation*} \normalsize
Finally, using the above result, the identification result for the collection of trials becomes \small
\begin{equation*}
    \begin{split}
        &\E[Y^{z,a} - Y^{z^\prime, a^\prime} | S = 0] =   \E \vast[ \E\Bigg [ \Bigg(  \dfrac{I(Z = z, A = a)}{\Pr[Z = z | X, S, I(S \in \mathcal S) = 1]\Pr[A = a| X, S, Z = z, L, I(S \in \mathcal S) = 1 ]} \\
&\quad  - \dfrac{I(Z = z^\prime, A = a^\prime)}{\Pr[Z = z^\prime | X, S, I(S \in \mathcal S) = 1]\Pr[A = a^\prime| X, S, Z = z^\prime, L, I(S \in \mathcal S) = 1 ]} \bigg) Y \Bigg | X , I(S \in \mathcal S) = 1 \Bigg] \Biggl| S = 0 \vast],
    \end{split}
\end{equation*} \normalsize
or equivalently, \small
\begin{equation*}
    \begin{split}
        &\E[Y^{z,a} - Y^{z^\prime, a^\prime} | S = 0] =  \dfrac{1}{\Pr[S = 0]}  \E \vast[ \Bigg(  \dfrac{I(Z = z, A = a)}{\Pr[Z = z | X, S, I(S \in \mathcal S) = 1]\Pr[A = a| X, S, Z = z, L, I(S \in \mathcal S) = 1 ]} \\
&\quad  - \dfrac{I(Z = z^\prime, A = a^\prime)}{\Pr[Z = z^\prime | X, S, I(S \in \mathcal S) = 1]\Pr[A = a^\prime| X, S, Z = z^\prime, L, I(S \in \mathcal S) = 1 ]} \bigg) \dfrac{I(S \in \mathcal S) Y \Pr[S = 0 | X]}{\Pr[I(S \in \mathcal S) = 1 | X]}  \vast].
    \end{split}
\end{equation*}
\normalsize

\clearpage
\section{Code to implement the estimators}\label{appendix:code}
We have provided \texttt{R} code that can be modified to implement the methods in a new dataset and a simulated dataset that illustrates the input data structure. The code is available at \url{https://github.com/serobertson/TransportingMultipleTrials}.


\ddmmyyyydate 
\newtimeformat{24h60m60s}{\twodigit{\THEHOUR}.\twodigit{\THEMINUTE}.32}
\settimeformat{24h60m60s}
\begin{center}
\vspace{\fill}\ \newline
\textcolor{black}{{\tiny $ $generalizability\_multiple\_trials, $ $ }
{\tiny $ $Date: \today~~ \currenttime $ $ }
{\tiny $ $Revision: \paperversionmajor.\paperversionminor $ $ }}
\end{center}

\end{document}